\documentclass[10pt]{amsart}

\usepackage{amssymb}
\usepackage{amsfonts,amsthm,amsmath}

\usepackage{graphicx,epstopdf}
\usepackage{setspace}
\usepackage{color}
\usepackage{xcolor}
\usepackage{multirow}
\usepackage{rotating}
\usepackage{float}
\usepackage{mathtools}

\def\E{{\mathbb{E}}}
\def\P{{\mathbb{P}}}

\def \I{{\mathbb{I}}}
\def \F{{\mathcal{F}}}

\newcommand{\dx}{\operatorname{d}\!}

\usepackage[english]{babel}
\usepackage[left=3cm,right=3cm,top=2.5cm,bottom=3cm]{geometry}

\usepackage{amsmath,amssymb,amsthm,bbm,color,graphics,version}
\usepackage{mathrsfs}

\newcommand{\bearno}{\begin{eqnarray*}}
\newcommand{\enarno}{\end{eqnarray*}}

 \def\Xi{X^{\infty}}

\def\P{{\mathbb P}}   
\def\E{{\mathbb E}}

\newtheorem{Thm}{Theorem}

\newtheorem{Lemma}{Lemma}
\newtheorem{Prop}{Proposition}

\newtheorem{Rem}{Remark} 
 
\title{Pricing American Options
	 Time-Capped by a Drawdown Event}

\author{Zbigniew Palmowski}
\address{Faculty of Pure and Applied Mathematics, Wroc\l aw University of Science and Technology, Wroc\l aw, Poland} \email{zbigniew.palmowski@pwr.edu.pl}
\author{Pawe\l\ St\c{e}pniak}
\address{Faculty of Pure and Applied Mathematics, Wroc\l aw University of Science and Technology, Wroc\l aw, Poland} \email{pawel.stepniak@pwr.edu.pl}

\thanks{
This work is partially supported by National Science Centre, Poland, under grants
No. 2021/41/B/HS4/00599.}

\date{\today}
\keywords{}

\begin{document}
\begin{abstract}
This paper presents a derivation of the explicit price for the perpetual American put option in the Black–Scholes model, time-capped by the first drawdown epoch beyond a predefined level. We demonstrate that the optimal exercise strategy involves executing the option when the asset price first falls below a specified threshold. The proof relies on martingale arguments and the fluctuation theory of L\'evy processes. To complement the theoretical findings, we provide numerical analysis.
\vspace{3mm}

\noindent {\sc Keywords.}
Black-Scholes model  $\star$ pricing $\star$ American option $\star$ optimal stopping $\star$ cap

\end{abstract}

\maketitle

\pagestyle{myheadings} \markboth{\sc P.\ St\c{e}pniak --- Z.\ Palmowski} {\sc Time-capped American options}

\vspace{1.8cm}

\tableofcontents

\newpage

\section{Introduction}
\subsection{Main results}

In mathematical finance, there has been growing interest in extending classical American options to include more complex conditions. Among these innovations are capped options, which aim to limit risk by introducing constraints on asset value or maturity time. These instruments appeal to investors due to their reduced liability and lower cost compared to standard options. Popularity of this type of financial instrument strongly depends though on
understanding their pricing, hedging, and optimal exercise policies. The desire to understand these features was our motivation for this paper.

The most common form of capped options involves applying a cap to the asset price, where the option is terminated if the asset value exceeds a specified threshold. One of the simplest examples of capped options are introduced in 1991 by the Chicago Board of Options Exchange European options written on the S\&P 100 and S\&P 500 with a cap on their payoff function (see \cite{cap}).
For other papers related to this type of cap see \cite{Zaevski2, Zaevski2b} and references therein.

Our study focuses on a non-deterministic time cap, where the termination of the option is triggered by the asset price experiencing a drawdown that exceeds a fixed threshold. Here, drawdown refers to the relative drop in the asset price compared to its historical maximum. We derive a closed-form formula for the price of the American put option under this random maturity condition and determine the corresponding optimal exercise strategy.

More formally, consider the asset price $S_t$ in the Black-Scholes market, that is, under the risk-neutral measure,
\begin{equation}\label{assetproce}
S_t=e^{X_t},
\end{equation}
where
\begin{equation}\label{X_def}
	X_t = x + \left( r - \frac{\sigma^2}{2} \right)t + \sigma B_t
\end{equation}
and $B_t$ is Brownian motion defined on a filtered probability space $(\Omega, \F, \{\mathcal{F}_t\}_{\{t\geq 0\}}, \P)$ with a natural filtration $\{\F_t\}_{\{t \geq 0\}}$ of $B_t$ satisfying usual conditions,
r is risk-free interest rate and $\sigma$ is a volatility.
By
\begin{equation}\label{define_overline_S}
	\overline{S}_t = e^{\overline{x}}\vee \sup_{0 \leq u \leq t} S_u
\end{equation}
we denote the running maximum of the asset price where
$e^{\overline{x}}$ is the historical maximum of the underlying asset price prior to the beginning of the contract, where $a\vee b = \max\{a,b\}$. Similarly, we will denote $a\wedge b = \min\{a,b\}$.
For the fixed threshold $c>0$ let
\begin{equation}\label{define_tau_D}
	\tau_D = \inf \left\{ t \geq 0: \frac{\overline{S}_t}{S_t} \geq e^c \right\}
\end{equation}
be the first time when the (relative) drawdown is bigger than $e^c$.
In the main result of this paper we identify the closed-form formula for the price of the American put option with the random maturity determined be a drawdown event given by
\begin{equation}\label{define_V}
	V(x, \overline{x}) = \sup_{\tau \in \mathcal{T}} \E_{x,\overline{x}} \left[e^{-r\tau\wedge\tau_D}(K - S_{\tau\wedge\tau_D})^+\right],
\end{equation}
for a family of stopping times $\mathcal{T}$ and fixed strike price $K > 0$.
Above, the subindex $x,\overline{x}$ attached to the expectation $\E_{x,\overline{x}}$ underlines the dependence of the mean on
the initial asset price $S_0=e^x$ and the historic (observed) supremum $\overline{S}_0=e^{\overline{x}}$. 
When $x=0$ we will skip this index. We restrict the domain of $V$ to $\mathcal{D} := \{ (x, \overline{x}) \in \mathbb{R}^2 : x \leq \overline{x}\}$.
In the main result we also find the optimal exercise rule, that is, a stopping rule $\tau^*$ such that
\begin{equation*}
	V(x, \overline{x}) = \E_{x,\overline{x}} [e^{-r\tau^*\wedge\tau_D}(K - S_{\tau^*\wedge\tau_D})^+].
\end{equation*}
To present the main result we first introduce so-called first scale function $W^{(r)}(x)$ of the process $X_t$ defined in \eqref{X_def}:
\begin{equation}\label{W_sum}
		W^{(r)}(x) = C e^{x} - C e^{\gamma x},
	\end{equation}
	where
\begin{equation}\label{gamma}
\gamma=-\frac{2r}{\sigma^2}
\end{equation}
and
	\begin{align}\label{C_def}
		C = \frac{1}{r + \frac{\sigma^2}{2}}.
	\end{align}
\begin{Thm}\label{thm2}
The optimal stopping barrier is the first downward asset price time
\begin{equation}\label{optimalstoppingrule}
	\tau^*=\inf\left\{t\geq 0: X_t\leq a^*\right\},
\end{equation}
where
	\begin{equation}\label{const_a}
	a^* = \log{\left( K\left(\frac{\gamma \left(e^{\gamma c} - e^c\right)}{(1-\gamma )e^c}\right)^{\frac{W^{(r)}(c)}{W^{(r)'}(c)}} \right)}.
\end{equation}
The option should be stopped at $X_t = a^*$ or $X_t = \overline{X_t}-c$, whichever occurs first. Moreover, the value function $V(x, \overline{x})$ is given explicitly in \eqref{valuecase1} (together with \eqref{V1}
-\eqref{V4} and \eqref{V5}) when $\overline{x} < a^* + c$
and in \eqref{valuefunction2} (together with \eqref{V6}-\eqref{V8}), otherwise.
\end{Thm}

\subsection{Literature overview}

Random termination is a common feature of many financial products. For example, the random termination moment could correspond to the default time of a company (see, e.g., \cite[p. 27]{BieleckiRutkowski} or \cite{Mladen}) or an asset-price-independent time cap following an exponential or Erlang distribution (see \cite{carr}). The latter case, which leads to the so-called Canadian approximation, has been analyzed in detail in \cite{Florin}. In some cases, the time cap is chosen to be unobservable; see \cite{gapeev1}.

Time-cap is closely related to cancellable options, which are terminated early when a specific event occurs. In fact, the price of the stopped option consists of the value of the cancellable option and the discounted payoff (under the risk-neutral measure) at the event time, provided the event happens before maturity.
Typically, this event is described as the first or last time when the underlying asset price
reaches a specific threshold; see, e.g., \cite{1st_pass, algo} and references therein.
Similar early termination features can be found in game options (like, e.g., Israeli options) where the seller has the right to terminate the contract early, subject to a fixed penalty paid to the buyer; see the seminal paper \cite{kifer} and subsequent works such as \cite{meyer, Zaevski2, Zaevski2b}.
Still, none of this works considers the drawdown event as a time-cap.

There are other studies addressing derivatives that closely resemble those analyzed in this paper.
Egloff, Farkas and Leippold in \cite{time_constraints} price American options with
stochastic stopping time constraints, where the exercise of the option is restricted to specific conditions being met.
These conditions, defined in terms of the states of a Markov process, are linked to a stochastic performance condition.
As the authors noted, such
performance-based constraints play an important role for structuring new
investment products and designing executive stock option plans with exercise conditions tied, for example, to outperforming a reference index.
To some extent, our motivation for considering a time-capped option is very similar:
we want to introduce an option that remains valid until is either terminated by the buyer or triggered by an event described above, whichever occurs first.
However, the difference in pricing is crucial. In \cite{time_constraints} the buyer can only exercise the option when a prespecified condition, associated with the performance of the underlying asset, is satisfied. Therefore, it was possible to transform the constrained pricing problem into an unconstrained optimal stopping problem that corresponds to a generalised barrier option pricing and a stochastic Cauchy-Dirichlet problem.
For the time-capped American option considered in this paper, we do not verify the drawdown event at the exercise time as it is in \cite{time_constraints}, but we exercise the option whenever the first drawdown event has occurred.

Several authors have also explored the concept of time-capping in American options.
In \cite{trabelsi}, the random cap is a~first hitting time of a fixed barrier by
the underlying asset price. In \cite{ott}, the fixed time cap is examined for
Russian and American look-back options. Finally, there are other related works.
For example, \cite{russian, random_put} investigate Russian options that terminate when the stock price hits its running maximum for the last time, as well as American options that terminate when the stock price reaches a prespecified level for the last time.

In our case, the American option terminates early if a drawdown event occurs. Specifically, this means the option is exercised either at a random stopping time or at the first time when the drawdown of the stock price exceeds a fixed threshold, whichever comes first.
Explicitly including protection against significant drawdowns in financial contracts is a common practice aimed at minimizing potential losses for the seller. Therefore, the list of papers addressing contracts that incorporate drawdown or drawup feature is quite long; see,
e.g. \cite{CZH, GZ, MA, ZP&JT2, ZP&JT, drawdownup2, PV, Sorn, Vec1,Vec2, drawdownup1, olympia}
and references therein.

\subsection{Organisation of the paper}
This paper is organised as follows. In the next section we give the proof of the main result. The proof is based on choosing a stopping rule and calculating the option value based on it. Next, a verification step via HJB system ensures that the chosen rule is indeed optimal and the achieved price is fair.
In Section \ref{sec:num} we present a numerical analysis.

\section{Proof of the main result}
In this section we prove Theorem \ref{thm2}.
We divide the proof into a few steps.

\subsection{Form of the stopping region}
The first crucial step is to prove that the optimal stopping rule for \eqref{define_V} is the first downward crossing time of a boundary that depends on the running supremum.
\begin{Prop}\label{thm1}
The optimal stopping time $\tau^*$ is of the following form:
\begin{equation}\label{11}
	\tau^* = \inf \{ t \geq 0: X_t \leq a_0(\overline{X_t}) \}
\end{equation}
for some function $a_0$.
\end{Prop}
\begin{proof}
	Let
\begin{equation}\label{stopregion}D = \{ (x,\overline{x})\in \mathbb{R}^2: V(x, \overline{x}) = (K - e^x)^+ \}
\end{equation}
be a stopping region, where the option should be exercised immediately.
Note that
\begin{equation}\label{Zt}Z_t = (X_t,\overline{X}_t)\end{equation}
is a Feller process.
By \cite[Thm. 2.7, p. 40 and (2.2.80), p. 49]{PS} it follows that $\tau^*=\inf\{t\geq 0: X_t\in D\}$.
	Suppose there exist $(x,\overline{x}) \in D$. We additionally assume that $x < \log K$, otherwise the immediate payout is zero. Let $\tau_y$ and $\tau_x$ be the optimal stopping rule for starting point $(y, \overline{x})$ and $(x, \overline{x})$ respectively. Observe that if $y < x < \log(K)$, then for a chosen $c$, we have $\tau_D(y, \overline{x}) \leq \tau_D(x,\overline{x})$, where $\tau_D(x,\overline{x})$ denotes $\tau_D$ for the starting point $(x, \overline{x})$. Further, we have:
	\begin{align*}
		V(y, \overline{x}) - V(x, \overline{x}) &= \E_{y,\overline{x}}[e^{-r\tau_y\wedge\tau_D(y,\overline{x})}(K - e^{X_{\tau_y\wedge\tau_D(y,\overline{x})}})^+] - \E_{x,\overline{x}}[e^{-r\tau_x\wedge\tau_D(x,\overline{x})}(K - e^{X_{\tau_x\wedge\tau_D(x,\overline{x})}})^+] \\
		&\leq \E_{y,\overline{x}}[e^{-r\tau_y\wedge\tau_D(y,\overline{x})}(K - e^{X_{\tau_y\wedge\tau_D(y,\overline{x})}})^+] - \E_{x,\overline{x}}[e^{-r\tau_y\wedge\tau_D(y,\overline{x})}(K - e^{X_{\tau_y\wedge\tau_D(y,\overline{x})}})^+] \\
		& = \E[e^{-r\tau_y\wedge\tau_D(y,\overline{x})}(K - e^{y + X_{\tau_y\wedge\tau_D(y,\overline{x})}})^+] - \E[e^{-r\tau_y\wedge\tau_D(y,\overline{x})}(K - e^{x + X_{\tau_y\wedge\tau_D(y,\overline{x})}})^+] \\
		& \leq \E[e^{-r\tau_y\wedge\tau_D(y,\overline{x})}(K - e^{y + X_{\tau_y\wedge\tau_D(y,\overline{x})}})] - \E[e^{-r\tau_y\wedge\tau_D(y,\overline{x})}(K - e^{x + X_{\tau_y\wedge\tau_D(y,\overline{x})}})] \\
		& = (e^x - e^y)\E[e^{-r\tau_y\wedge\tau_D(y,\overline{x}) + X_{\tau_y\wedge\tau_D(y,\overline{x})}}] = e^x - e^y.
	\end{align*}
Therefore we get
\begin{equation}
		V(y, \overline{x}) - V(x, \overline{x}) \leq  e^x - e^y = (K - e^y) - (K-e^x)
\end{equation}
and further
\begin{equation}
		V(y, \overline{x}) \leq (K - e^y) \leq  (K - e^y)^+.
\end{equation}
On the other hand the payoff function of the option cannot be higher than its value function, therefore
\begin{equation}
	V(y, \overline{x}) \geq  (K - e^y)^+.
\end{equation}
This gives $(y, \overline{x}) \in D$. This leads to the conclusion, that for a certain $\overline{x}$ the optimal stopping region can be achieved by the pair $(X_t, \overline{X_t})$ when $X_t$ drops down to some value $a_0(\overline{X_t})$ before it reaches its past maximum.

The question remains whether the optimal stopping region can be also reached from below, when both $X_t$ and $\overline{X_t}$ hit a certain level for the first time. We will show that this scenario is not possible.
Indeed, assume a contrario that there exists a threshold $b$ such that there exist $x$ and  $\overline{x} < b$ satisfying $(x, \overline{x}) \notin D$ and $(b,b) \in D$. Let us take a positive $\varepsilon$ such that $b - \varepsilon > \overline{x}$. Let $\theta$ be the first time when process $X_t$ reaches the level $b - \varepsilon$ from below, i.e.
\begin{equation}
	\theta = \inf \{ t > 0: X_t = b - \varepsilon \}
\end{equation}
and let $\tau_b$ be the first time when process $X_t$  reaches the level $b$ from below. Clearly, we have $\theta < \tau_b$ and $e^{X_\theta} < e^{X_{\tau_b}}.$ Therefore
\begin{equation}\label{contradiction}
	\E_{x,\overline{x}}[e^{-r\theta\wedge\tau_D(x,\overline{x})}(K - e^{X_{\theta\wedge\tau_D(x,\overline{x})}})^+] \geq \E_{x,\overline{x}}[e^{-r\tau_b\wedge\tau_D(x,\overline{x})}(K - e^{X_{\tau_b\wedge\tau_D(x,\overline{x})}})^+].
\end{equation}
This implies it must dominate any alternative stopping rule in terms of the expected payoff. Inequality (\ref{contradiction}) contradicts this assumption, which completes the proof.

\end{proof}

\subsection{HJB and verification step}

By $\mathcal{L}$ we denote an infinitesimal generator of the Markov process
$Z_t$ defined in \eqref{Zt}, which equals
\begin{align}
	&\mathcal{L}f(x, \overline{x}) = \left(r - \frac{\sigma^2}{2}\right) \frac{\partial}{\partial x} f(z) + \frac{\sigma^2}{2}\frac{\partial^2}{\partial x^2} f(z) \quad \text{for}\quad 0 < x < \overline{x}
\end{align}
and the domain of this (full) generator includes the functions $f\in \mathcal{C}^2_0(\mathbb{R})$ such that
\begin{align}
	& \frac{\partial}{\partial \overline{x}} f(x, \overline{x})  =  0  \quad \text{for}\quad  x = \overline{x}. \label{gen_dom_cut}
\end{align}
In the next step we prove the following verification lemma.
\begin{Lemma}\label{HJB}
Let  $\hat{V}(x, \overline{x}): \mathbb{R}^2 \rightarrow \mathbb{R}$ be a function defined on $\mathcal{D} := \{ (x, \overline{x}) \in \mathbb{R}^2 : x \leq \overline{x}\}$.
Assume that $\hat{V}(x, \overline{x}) \in \mathcal{C}^2_0(\mathbb{R})$ and that it fulfills condition \eqref{gen_dom_cut}.
Assume that for some function $b$,
\begin{align}
	(\mathcal{L} \hat{V} - r\hat{V})(x, \overline{x}) &= 0\quad \text{for}\ x> b(\overline{x}), \label{HJB1} \\
	(\mathcal{L} \hat{V} - r\hat{V})(x, \overline{x}) &\leq 0\quad \text{for}\ x\leq  b(\overline{x}),\label{HJB1b}
\end{align}
\begin{align}
	\hat{V}(x, \overline{x}) &= (K - e^x)^+\quad \text{for}\ x \leq b(\overline{x}), \\
	\hat{V}(x, \overline{x}) &> (K - e^x)^+\quad \text{for}\ x> b(\overline{x}),
\end{align}
\begin{align}
	\hat{V}(x, \overline{x}) \big|_{x = b(\overline{x})} &= (K - e^{b(\overline{x})})^+, \label{HJB6}  \\
		\frac{\partial}{\partial x} \hat{V}(x, \overline{x}) \big|_{x = b(\overline{x})} &= \frac{\partial}{\partial x} (K - e^x)^+ \big|_{x = b(\overline{x})} \quad \text{if}\ b(\overline{x})<\overline{x}-c.\label{HJB7}
\end{align}
Then $\hat{V}(x, \overline{x})\geq V(x, \overline{x})$.
\end{Lemma}
\begin{Rem}
\rm Conditions \eqref{HJB6} and \eqref{HJB7} are the so-called smooth paste conditions of the value function. Note that the condition \eqref{HJB6} is mainly required
to write \eqref{HJB7} which is used in the proof of Lemma \ref{HJB}.
\end{Rem}
\begin{proof}
Due to the assumed smoothness of $\hat{V}$, the smooth paste conditions \eqref{HJB6} - \eqref{HJB7} and an appropriate version It\^{o}'s theorem
(see \cite[p. 208]{EisenbaumKyprianou}) we have
\begin{align}
	&e^{-rt}\hat{V}(X_t, \overline{X}_t) =
\hat{V}(x,\overline{x}) + \sigma\int_0^t e^{-r u} \frac{\partial}{\partial x} \hat{V}(X_u, \overline{X}_u) \dx B_u + \int_0^t e^{-r u} \frac{\partial}{\partial \overline{x}} \hat{V}(X_u, \overline{X}_u) \dx \overline{X}_u \nonumber\\
&\quad + \int_0^t e^{-r u} \left( \mathcal{L}\hat{V}(X_u, \overline{X}_u) - r\hat{V}(X_u, \overline{X}_u) \right) \dx u\\&\quad  +\frac{1}{2}
\int_0^t \left( \frac{\partial}{\partial x} \hat{V}(x, \overline{x}) \big|_{x = b(\overline{x})} - \frac{\partial}{\partial x} (K - e^x)^+ \big|_{x = b(\overline{x})}\right)\dx L(s), \nonumber
\end{align}
where $L$ is a local time of the process $X - b(\overline{X})$ at 0. Observe that $\dx \overline{X}_u = \I\{ X_u = \overline{X}_u \}\dx X_u$.

Now, requirement (\ref{gen_dom_cut}) guarantees that the integral over $\dx \overline{X}_u$ is zero. Similarly, the smooth-paste conditions make sure that the integral over the local time also vanishes. Finally, relations (\ref{HJB1}) and (\ref{HJB1b}) lead to the conclusion that the integral over $\dx u$ is non-positive. If we take the expectation of both sides, we get the following result
\begin{align*}
	e^{-rt}&\E_{x,\overline{x}}\hat{V}(X_t, \overline{X}_t) = \hat{V}(x,\overline{x}) + \sigma \E_{x,\overline{x}} \int_0^t e^{-r u} \frac{\partial}{\partial x} \hat{V}(X_u, \overline{X}_u) \dx B_u \\
	 &+  \E_{x,\overline{x}} \int_0^t e^{-r u} \left( \mathcal{L}\hat{V}(X_u, \overline{X}_u) - r\hat{V}(X_u, \overline{X}_u) \right) \dx u \leq \hat{V}(x,\overline{x})
\end{align*}
since the integral over Brownian motion is a zero-mean local martingale. Hence by made assumptions, the process  $e^{-rt\wedge \tau_D}\hat{V}(X_{t\wedge \tau_D}, \overline{X}_{t\wedge \tau_D})=e^{-rt\wedge \tau_D}\hat{V}(Z_{t\wedge \tau_D})$ is a supermartingale and $\hat{V}(x,\overline{x})$ is a superharmonic function that dominates the payout. Now, from \cite[(2.2.80), p. 49]{PS} we additionally know that $\hat{V}$ is also lower semi-continuous. It allows us to use \cite[Thm. 2.7, p. 40]{PS} and claim that it is the optimal solution to the considered stopping problem. Hence $\hat{V}(x, \overline{x})\geq V(x, \overline{x})$.
\end{proof}
\begin{Rem}\label{uwaga2}
By Proposition \ref{thm1} we know that the optimal stopping rule $\tau^*$
is of the one-sided form. In the next step, we postulate that the optimal
stopping boundary is even more specific, namely, that
\begin{itemize}
\item $b(\overline{x})=a^*$ where $a^*$ is defined in \eqref{const_a} when $\overline{x} < a^* + c$;
\item $b(\overline{x})=\overline{x}-c$ when $a^* + c <  \overline{x} < \log(K)+c$
\end{itemize}
for $a^*$ defined formally in \eqref{const_a}.
We will calculate the value function
\begin{equation*}
\hat{V}(x, \overline{x})=\E_{x,\overline{x}} [e^{-r\tau^*\wedge\tau_D}(K - S_{\tau^*\wedge\tau_D})^+]
\end{equation*}
for this postulated stopping rule  $\tau^*$ and we will show that all the assumptions of Lemma \ref{HJB}
are satisfied for $\hat{V}$. Hence in this case
$\hat{V}(x, \overline{x})\geq V(x, \overline{x})$  by  Lemma \ref{HJB} and $\hat{V}(x, \overline{x})\leq V(x, \overline{x})$ due to the fact that
we choose specific stopping rule. Thus $\hat{V}(x, \overline{x})=V(x, \overline{x})$ is a true value function
and $\tau^*$ is the optimal stopping rule.

From general stopping theory applied to the Markov process $(S_{t\wedge \tau_D}, \overline{S}_{t\wedge \tau_D})$
(see \cite[Thm. 2.7, p. 40 and (2.2.80), p. 49]{PS}) it follows that the stopping region is the set when the value function meets the payout function
and hence it is unique which is due to the existence of the value function
$V(x, \overline{x})$. In other words, our stopping region is unique as well.
\end{Rem}

\subsection{Identifcation of the value function}
According to Remark \ref{uwaga2}, we
choose
\begin{itemize}
\item $b(\overline{x})=a^*$ where $a^*$ is defined in \eqref{const_a} when $\overline{x} < a^* + c$;
\item $b(\overline{x})=\overline{x}-c$ when $a^* + c <  \overline{x} < \log(K)+c$,
\end{itemize}
and we will find now the value function $\hat{V}(x, \overline{x})$ for this stopping time. Later we will verify that
the function
\begin{equation*}
\hat{V}(x, \overline{x})=\E_{x,\overline{x}} [e^{-r\tau^*\wedge\tau_D}(K - S_{\tau^*\wedge\tau_D})^+]
\end{equation*}
satisfies all assumptions of the verification lemma.

Observe first
that when $\overline{x} \geq  \log(K)+c$ then the option is immediately exercised.
In this case assertion of Theorem \ref{thm2} is trivial and holds true.

When the second case hold, that is, when
$ a^* + c \leq   \overline{x} < \log(K)+c$, then the level $a^*$ cannot
be achieved by the process $X_t$ since the drawdown event will happen first.
Therefore we choose $b(\overline{x})=\overline{x}-c$ to verify all conditions of Lemma \ref{HJB}.

To identify the value function $\hat{V}(x, \overline{x})$, we will need some fluctuation identities and so-called scale functions.
In particular, we will need	

the first scale function which for $r \geq 0$ is defined as a continuous function $W^{(r)}$ on non-negative half-line with a Laplace transform
	$\int_0^\infty e^{-\theta x}W^{(r)}(x)\dx x = \frac{1}{\left(r - \frac{\sigma^2}{2}\right)\theta + \frac{\sigma^2\theta^2}{2}-r}$  for sufficiently large $\theta$. 
Hence $W^{(r)}(x)$ is given in \eqref{W_sum}.

With the first scale function we associate the second one given by
	\begin{equation*}
	Z^{(r)}(x) = 1 + r\int_{0}^{x} W^{(r)}(y)\dx y.
	\end{equation*}
By direct checking we can verify that
\begin{equation}\label{martingalescale}
\mathcal{L}W^{(r)}(x)-rW^{(r)}(x)=0 \quad \text{and}\quad  \mathcal{L}Z^{(r)}(x)-rZ^{(r)}(x)=0.
\end{equation} 	

\subsubsection{The case when $\overline{x} < a^* + c$}	
We start from the case when $\overline{x} < a^* + c$, that is, when $b(\overline{x})=a^*$ in Lemma \ref{HJB} with $a^*$ defined in \eqref{const_a}. Let $\tau = \tau^* \wedge \tau_D$ be the time then the option is stopped either by downward crossing level $a^*$ or by the drawdown event. Observe that in this case
	\begin{align}\label{V_div2}
		\hat{V}(x, \overline{x}) &=  \E_{x, \overline{x}}\left[e^{-r\tau}\left(K - e^{X_{\tau}}\right)^+\right] = \E_{x, \overline{x}}\left[e^{-r\tau}\left(K - e^{X_{\tau}}\right)^+\I\{ \tau^-_{a^*\vee \overline{x}-c} < \tau^+_{\overline{x}} \} \right]
		\\
		&\quad + \E_{x, \overline{x}}\left[e^{-r\tau}\left(K - e^{X_{\tau}}\right)^+\I\{ \tau^-_{a^*\vee \overline{x}-c} > \tau^+_{\overline{x}} \} \right] = \E_{x, \overline{x}}\left[e^{-r\tau^-_{a^*}}\left(K - e^{a^*}\right)\I\{ \tau^-_{a^*} < \tau^+_{\overline{x}} \} \right]
		\nonumber\\
		&\quad + \E_{x, \overline{x}}\left[ e^{-r\tau^+_{\overline{x}}} \I \{ \tau^+_{\overline{x}} < \tau^-_{a^*} \} \right] \E_{\overline{x}, \overline{x}}\left[e^{-r\tau}\left(K - e^{X_{\tau}}\right)^+\right]
		\nonumber\\
		&=  \E_{x, \overline{x}}\left[e^{-r\tau^-_{a^*}}\left(K - e^{a^*}\right)\I\{ \tau^-_{a^*} < \tau^+_{\overline{x}} \} \right]  + \E_{x, \overline{x}}\left[ e^{-r\tau^+_{\overline{x}}} \I \{ \tau^+_{\overline{x}} < \tau^-_{a^*} \} \right]
		\nonumber\\
		& \quad \times \left(\E_{\overline{x}, \overline{x}}\left[e^{-r\tau^-_{a^*}}\left(K - e^{a^*}\right) \I\{ \tau^-_{a^*} < \tau^+_{a^*+c} \} \right]  + \right.
		\nonumber\\
		&\quad + \left.  \E_{x, \overline{x}}\left[ e^{-r\tau^+_{a^*+c}} \I\{ \tau^-_{a^*} > \tau^+_{a^*+c} \} \right] \E_{a^*+c, a^*+c}\left[e^{-r\tau_D}\left(K - e^{X_{\tau_D}}\right) \I\{ \tau_D < \tau^+_{\log K + c} \} \right] \right),
		\nonumber\\
&:= V_1(x, \overline{x})+V_2(x, \overline{x})\Bigl(V_3(\overline{x}) +V_4(\overline{x})V_5\Bigr),\label{valuecase1}
	\end{align}
	where
\[\tau_x^- = \inf\{ t \geq 0:\  X_t \leq x \} \quad \text{and}\quad \tau_x^+ = \inf\{ t \geq 0:\ X_t \geq x \}.\]
	The condition $\tau_D < \tau^+_{\log K + c}$ in $V_5$ is necessary to ensure that $K \geq e^{X_{\tau_D}}$ so that $\left(K - e^{X_{\tau_D}}\right)^+ = \left(K - e^{X_{\tau_D}}\right)$.
	We will analyse above components one by one. Using formula \cite[eq. (2.4)]{MP} we get that
	\begin{equation}\label{V1}
		V_1(x, \overline{x}) = \E_{x, \overline{x}}\left[e^{-r\tau}\left(K - e^{a^*}\right)\I\{ \tau^-_{a^*} < \tau^+_{\overline{x}} \} \right] = \left(K - e^{a^*}\right)\left( Z^{(r)}(x-a^*) - Z^{(r)}(\overline{x}-a^*)\frac{W^{(r)}(x-a^*)}{W^{(r)}(\overline{x}-a^*)} \right)
	\end{equation}
	and
	\begin{equation}\label{V3}
		V_3(\overline{x}) = \E_{\overline{x}, \overline{x}}\left[e^{-r\tau^-_{a^*}}\left(K - e^{a^*}\right) \I\{ \tau^-_{a^*} < \tau^+_{a^*+c} \} \right] = \left(K - e^{a^*}\right)\left( Z^{(r)}(\overline{x}-a^*) - Z^{(r)}(c)\frac{W^{(r)}(\overline{x}-a^*)}{W^{(r)}(c)} \right).
	\end{equation}
	From \cite[eq. (2.3)]{MP} we also have that
	\begin{equation}\label{V2}
		V_2(x, \overline{x}) = \E_{x, \overline{x}}\left[ e^{-r\tau^+_{\overline{x}}} \I \{ \tau^+_{\overline{x}} < \tau^-_{a^*} \} \right] = \frac{W^{(r)}(x-a^*)}{W^{(r)}(\overline{x}-a^*)}
	\end{equation}
	and
	\begin{equation}\label{V4}
		V_4(\overline{x}) = \E_{x, \overline{x}}\left[ e^{-r\tau^+_{a^*+c}} \I\{ \tau^-_{a^*} > \tau^+_{a^*+c} \} \right] = \frac{W^{(r)}(\overline{x}-a^*)}{W^{(r)}(c)}.
	\end{equation}
	To find the last component $V_5=\E_{a^*+c, a^*+c}\left[e^{-r\tau_D}\left(K - e^{X_{\tau_D}}\right) \I\{ \tau_D < \tau^+_{\log K + c} \} \right]$
	we introduce the following notations:
	\begin{align}
		&\lambda(d,r) = \frac{W^{(r)'}(d)}{W^{(r)}(d)}, \\ 
		&F_{r,d}(y) = \lambda(d,r)e^{-y\lambda(d,r)},\quad y \in \mathbb{R}_+, \label{v5a}\\
		&\Delta(d,r) = \frac{\sigma^2}{2}\left[ W^{(r)'}(d) - \lambda(d,r)^{-1}W^{(r)''}(d) \right].
	\end{align}
	Now, by \cite[eq. (3.11)]{MP} we obtain
	\begin{align*}
		V_5 = &\E_{a^*+c, a^*+c}\left[e^{-r\tau_D}\left(K - e^{X_{\tau_D}}\right) \I\{ \tau_D < \tau^+_{\log K + c} \} \right] = K\E_{a^*+c, a^*+c} \left[e^{-r\tau_D}\I\{ \tau_D < \tau^+_{\log K + c} \} \right]
		\\
		&- \E_{a^*+c, a^*+c} \left[e^{-r\tau_D + X_{\tau_D}}\I\{ \tau_D < \tau^+_{\log K + c} \} \right] = K\int_{a^*+c}^{\log K +c} F_{r,c}(v - (a^*+c))\Delta(c,r) \dx v
		\nonumber\\
		&- \int_{a^*+c}^{\log K +c}e^{v-c} F_{r,c}(v - (a^*+c))\Delta(c,r) \dx v.		\nonumber
	\end{align*}
Plugging the expression for $F_{r,d}(y)$ in \eqref{v5a} we derive
	\begin{align}
		V_5 =&= \Delta(c,r) e^{(a^*+c)\lambda(c,r)}\left( K\left(e^{-(a^*+c)\lambda(c,r)} - e^{-(\log K+c)\lambda(c,r)}\right) \right.
		\nonumber\\
		&- \left. e^{-c}\frac{\lambda(c,r)}{1-\lambda(c,r)}\left( e^{(1-\lambda(c,r))(\log(K)+c)} - e^{(1-\lambda(c,r))(a^*+c)} \right) \right) \nonumber\\
		&= \Delta(c,r)\left( K\left(1 - \frac{e^{-\lambda(c,r)\left(\log(K) - a^*\right)}}{1-\lambda(c,r)}\right) + \frac{\lambda(c,r)e^{a^*}}{1-\lambda(c,r)} \right). \label{V5}
	\end{align}

Inserting formulas (\ref{V1})-(\ref{V4}) and (\ref{V5}) into \eqref{V_div2} gives the price $\hat{V}(x, \overline{x})$
which is in $\mathcal{C}^2_0$ and it is a linear combination of the scale functions $W$ and $Z$ and therefore by \eqref{martingalescale} condition we have that
$\mathcal{L} \hat{V} - r\hat{V}(x, \overline{x}) = 0$ and conditions \eqref{HJB1} and \eqref{HJB1b} are satisfied.

We recall that to have $\hat{V}(x, \overline{x})$ in the domain of the generator $\mathcal{L}$ of the Markov process $Z_t$ (see \eqref{gen_dom_cut}) we need
	\begin{equation}
		\frac{\partial }{\partial \overline{x}} \hat{V} (x, \overline{x}) \big|_{x = \overline{x}} = 0
	\end{equation}
	which we will show now. Observe that
	\begin{align*}
		\frac{\partial }{\partial \overline{x}} \hat{V} (x, \overline{x}) = &\frac{\partial }{\partial \overline{x}} V_1(x, \overline{x}) + V_5\left( V_4(\overline{x})
\frac{\partial }{\partial \overline{x}} V_2(x, \overline{x}) + V_2(x, \overline{x}) \frac{\partial }{\partial \overline{x}} V_4(\overline{x}) \right)
\\& + V_2(x, \overline{x}) \frac{\partial }{\partial \overline{x}} V_3(\overline{x}) + V_3(\overline{x}) \frac{\partial }{\partial \overline{x}} V_2(x, \overline{x}).
	\end{align*}
Note that $V_4(\overline{x}) V_2(x, \overline{x})$ does not depend on
 $\overline{x}$.  Therefore we have $ V_4(\overline{x}) \frac{\partial }{\partial \overline{x}} V_2(x, \overline{x}) + V_2(x, \overline{x}) \frac{\partial }{\partial \overline{x}} V_4(\overline{x}) = 0$.
Hence	
\begin{align}
			&\frac{\partial }{\partial \overline{x}} \hat{V}(x, \overline{x}) \big|_{x = \overline{x}} = -\left( K - e^{a^*} \right)\frac{Z^{(r)'}(\overline{x}-a)W^{(r)}(\overline{x}-a) - W^{(r)'}(\overline{x}-a)Z^{(r)}(\overline{x}-a)}{W^{(r)}(\overline{x}-a)} \\
			\quad & - \frac{W^{(r)'}(\overline{x}-a)}{W^{(r)}(\overline{x}-a)}\left( K - e^{a^*}\right)\left( Z^{(r)}(\overline{x}-a) - \frac{Z^{(r)}(c)}{W^{(r)}(c)}W^{(r)}(\overline{x}-a) \right) \nonumber\\
			\quad & + \left( K - e^{a^*}\right) \left( Z^{(r)}(\overline{x}-a) - \frac{Z^{(r)'}(c)}{W^{(r)}(c)}W^{(r)'}(\overline{x}-a) \right). \nonumber
	\end{align}
Now, simple algebra and simplifications produce
\begin{align}
\frac{\partial }{\partial \overline{x}} \hat{V}(x, \overline{x}) \big|_{x = \overline{x}}
			 = \frac{K - e^{a^*}}{W^{(r)}(\overline{x}-a)} \left( W^{(r)'}(\overline{x}-a)Z^{(r)}(\overline{x}-a) - Z^{(r)'}(\overline{x}-a)W^{(r)}(\overline{x}-a) - W^{(r)'}(\overline{x}-a)Z^{(r)}(\overline{x}-a)  \right. \nonumber\\
		\quad  + \left. \frac{Z^{(r)}(c)}{W^{(r)}(c)}W^{(r)}(\overline{x}-a)W^{(r)'}(\overline{x}-a) + Z^{(r)'}(\overline{x}-a)W^{(r)}(\overline{x}-a) -  \frac{Z^{(r)}(c)}{W^{(r)}(c)}W^{(r)}(\overline{x}-a)W^{(r)'}(\overline{x}-a) \right) = 0\nonumber
	\end{align}
which gives \eqref{gen_dom_cut}.

We recall that in the first case $b(\overline{x})=a^*$. Observe that $W^{(r)}(0) = 0$ and $Z^{(r)}(0) = 1$. Because of that, 	$V_1(a^*, \overline{x}) = \left(K - e^{a^*}\right) $ and $V_2(a^*, \overline{x}) = 0$. This immediately gives (\ref{HJB6}). We will show that the smooth paste condition (\ref{HJB7}) holds true, that is
	\begin{equation}
		\frac{\partial}{\partial x} \hat{V}(x, \overline{x}) \big|_{x = a^*} =  -e^{a^*}.
	\end{equation}
	Observe, that out of factors $V_1$ to $V_5$ only $V_1$ and $V_2$ are dependent on $x$. Therefore
	\begin{equation}\label{raza}
		\frac{\partial}{\partial x} \hat{V}(x, \overline{x}) \big|_{x = a^*} = \frac{\partial}{\partial x} V_1(x, \overline{x}) \big|_{x = a^*} + (V_3(\overline{x}) + V_4(\overline{x}) V_5) \frac{\partial}{\partial x} V_2(x, \overline{x}) \big|_{x = a^*}.
	\end{equation}
	Direct computations give
	\begin{equation}\label{razb}
		\frac{\partial}{\partial x} V_1(x, \overline{x}) \big|_{x = a^*} = \left( K - e^{a^*} \right)\left( -(1-\gamma)C \frac{Z^{(r)}(\overline{x}-a^*)}{W^{(r)}(\overline{x}-a^*)}  \right)
	\end{equation}
	and
		\begin{equation}\label{razc}
		\frac{\partial}{\partial x} V_2(x, \overline{x}) \big|_{x = a^*} =  \frac{(1-\gamma)C}{W(\overline{x}-a^*)}.
	\end{equation}
	Further, from \eqref{raza} and using \eqref{razb} and \eqref{razc} we can conclude that
	\begin{align}
		&\frac{\partial}{\partial x} \hat{V}(x, \overline{x}) \big|_{x = a^*} = \frac{(1-\gamma)C}{W(\overline{x}-a^*)} \left[ -\left( K - e^{a^*} \right)Z^{(r)}(\overline{x}-a^*) \right. \\
		\nonumber  &+ \left. \left(K - e^{a^*}\right)\left( Z^{(r)}(\overline{x}-a^*) - Z^{(r)}(c)\frac{W^{(r)}(\overline{x}-a^*)}{W^{(r)}(c)} \right) +  V_4(\overline{x}) V_5 \right].
\end{align}	
Now, we plug the expression $V_4(\overline{x})$ given in \eqref{V4}:
\begin{align}
		\nonumber 	\frac{\partial}{\partial x} \hat{V}(x, \overline{x}) \big|_{x = a^*}&= \frac{(1-\gamma)C}{W(\overline{x}-a^*)}\left[ -\left( K - e^{a^*} \right)Z^{(r)}(c)\frac{W^{(r)}(\overline{x}-a^*)}{W^{(r)}(c)} + \frac{W^{(r)}(\overline{x}-a^*)}{W^{(r)}(c)} V_5 \right] \\
\end{align}	
and then the expression $V_5$ given in \eqref{V5}:
\begin{align}
		\nonumber 	\frac{\partial}{\partial x} \hat{V}(x, \overline{x}) \big|_{x = a^*}&= \frac{(1-\gamma)C}{W(c)}\left[ \Delta(c,r)\left( K\left(1 - \frac{e^{-\lambda(c,r)\left(\log(K) - a^*\right)}}{1-\lambda(c,r)}\right) + \frac{\lambda(c,r)e^{a^*}}{1-\lambda(c,r)} \right) - Z^{(r)}(c) \left(K - e^{a^*}\right) \right],
\end{align}	
Finally, after simple algebra we have
\begin{align}
		\nonumber
\frac{\partial}{\partial x} \hat{V}(x, \overline{x}) \big|_{x = a^*}&= \frac{1-\gamma}{e^c - e^{\gamma c}} \left[ \frac{\Delta(c,r)}{1-\lambda(c,r)} \left( K\left( 1 - \lambda(c,r) - e^{-\lambda(c,r)\left(\log(K) - a^*\right)} \right) + \lambda(c,r)e^{a^*} \right) - \frac{e^{\gamma c}-\gamma e^c}{1-\gamma}\left(K - e^{a^*}\right) \right] \\
		\nonumber &= \frac{1-\gamma}{e^c - e^{\gamma c}} \left[ -\left(K - e^{a^*}\right) \left( \frac{e^{\gamma c}-\gamma e^c}{1-\gamma} + \frac{\lambda(c,r)\Delta(c,r)}{1-\lambda(c,r)} \right) + \frac{K\Delta(c,r)}{1-\lambda(c,r)}\left(1 - e^{-\lambda(c,r)\left(\log(K) - a^*\right)} \right) \right].
	\end{align}	
Observe that
	\begin{equation}
		\frac{\lambda(c,r)\Delta(c,r)}{1-\lambda(c,r)} = -e^{c}.
	\end{equation}
	Therefore:
	\begin{align}
			&\frac{\partial}{\partial x} \hat{V}(x, \overline{x}) \big|_{x = a^*} = \frac{1-\gamma}{e^c - e^{\gamma c}} \left[ \left(K - e^{a^*}\right) \left( e^c - \frac{e^{\gamma c}-\gamma e^c}{1-\gamma} +  \right) + \frac{Ke^c}{\lambda(c,r)}\left(1 - e^{-\lambda(c,r)\left(\log(K) - a^*\right)} \right) \right] \\
			\nonumber & = \frac{1-\gamma}{e^c - e^{\gamma c}} \left[ \left(K - e^{a^*}\right)\frac{e^c - e^{\gamma c}}{1-\gamma} - Ke^c \frac{e^c - e^{\gamma c}}{e^c - \gamma e^{\gamma c}} \left(1 - e^{-\lambda(c,r)\left(\log(K) - a^*\right)} \right)\right]\\
			\nonumber & = -e^{a^*} + K\left[ 1 - \frac{(1-\gamma)e^c}{e^c - \gamma e^{\gamma c}} \left( 1 - e^{ -\lambda(c,r)(\log(K) - a^*) } \right) \right].
	\end{align}
Now,  $a^*$ given in \eqref{const_a} is chosen in such a way that
	\begin{equation*}
		1 - \frac{(1-\gamma)e^c}{e^c - \gamma e^{\gamma c}} \left( 1 - e^{ -\lambda(c,r)(\log(K) - a^*) } \right) = 0
	\end{equation*}
which completes the proof of the smooth paste condition.	

This completes the proof of the main result in the first case.

\subsection{The case of $ a^* + c <  \overline{x} < \log(K)+c$}	
We recall that in this case we choose $b(\overline{x})=\overline{x}-c$ in Lemma \ref{HJB}.
Furthermore, we have
		\begin{align}\label{V_div3}
		\hat{V}(x, \overline{x}) &=  \E_{x, \overline{x}}\left[ e^{-r\tau^-_{\overline{x}-c}}\left( K - e^{X_{\tau^-_{\overline{x}-c}}} \right) \I\{ \tau^-_{\overline{x}-c} < \tau^+_{\overline{x}} \} \right] \\
		& \quad + \E_{x, \overline{x}}\left[ e^{-r\tau^+_{\overline{x}}}  \I\{ \tau^-_{\overline{x}-c} > \tau^+_{\overline{x}} \}\right]\E_{\overline{x}, \overline{x}}\left[ e^{-r\tau_D}\left( K - e^{X_{\tau_D}} \right)  \I\{ \tau_D < \tau^+_{\log(K)+c} \} \right] \nonumber
\\&:= V_6(x, \overline{x})+V_7(x, \overline{x})V_8(\overline{x}).\label{valuefunction2}
	\end{align}
Now we use \cite[equations (2.3), (2.4), (3.11)]{MP} to get
\begin{align}
	V_6(x, \overline{x}) &= \E_{x, \overline{x}}\left[ e^{-r\tau_D}\left( K - e^{X_{\tau_D}} \right) \I\{ \tau^-_{\overline{x}-c} < \tau^+_{\overline{x}} \} \right] \nonumber\\&= \left( K - e^{\overline{x}-c} \right) \left[ Z^{(r)}(x+c-\overline{x}) - \frac{Z^{(r)}(c)}{W^{(r)}(c)}W^{(r)}(x+c-\overline{x}) \right],
\label{V6}
\\
	V_7(x, \overline{x}) &= \E_{x, \overline{x}}\left[ e^{-r\tau^+_{\overline{x}}}  \I\{ \tau^-_{\overline{x}-c} > \tau^+_{\overline{x}} \}\right] =\frac{W^{(r)}(x+c-\overline{x})}{W^{(r)}(c)}, \label{V7}\\
	V_8(\overline{x}) &= \E_{\overline{x}, \overline{x}}\left[ e^{-r\tau_D}\left( K - e^{X_{\tau_D}} \right)  \I\{ \tau_D < \tau^+_{\log(K)+c} \} \right] \nonumber\\&= \Delta(c,r)K - e^{\overline{x}} + Ke^{\lambda(c,r)(\overline{x} - \log(K) - c)}\left( e^c - \Delta(c,r) \right).\label{V8}
\end{align}
Observe that in this case $\hat{V}(x, \overline{x})$
is again in $\mathcal{C}^2_0$ and it is a linear combination of the scale functions $W$ and $Z$ and therefore by \eqref{martingalescale} we have
$(\mathcal{L} \hat{V} - r\hat{V})(x, \overline{x}) = 0$ and conditions \eqref{HJB1} and \eqref{HJB1b} are satisfied.
To have \eqref{gen_dom_cut} satisfied, we need to verify that
$	\frac{\partial}{\partial \overline{x}} \hat{V}(x, \overline{x}) \big|_{x = \overline{x}} = 0$. Note that
\begin{align}
	&\frac{\partial}{\partial \overline{x}} \hat{V}(x, \overline{x}) \big|_{x = \overline{x}} = \left. \left(\frac{\partial}{\partial \overline{x}} V_6(x, \overline{x})  +  V_7(x,\overline{x})\frac{\partial}{\partial \overline{x}} V_8( \overline{x}) + V_8(\overline{x})\frac{\partial}{\partial \overline{x}} V_7(x, \overline{x}) \right) \right|_{x = \overline{x}} \\
	& = \left( K - e^{\overline{x}-c} \right) \left[ -Z^{(r)'}(c) + \frac{Z^{(r)}(c)}{W^{(r)}(c)}W^{(r)'}(c) \right] - e^{\overline{x}} + K\lambda(c,r)e^{\lambda(c,r)(\overline{x} - \log(K) - c)}\left( e^c - \Delta(c,r) \right) \nonumber\\
	&- \frac{W^{(r)'}(c)}{W^{(r)}(c)}\left( \Delta(c,r)K - e^{\overline{x}} + Ke^{\lambda(c,r)(\overline{x} - \log(K) - c)}\left( e^c - \Delta(c,r) \right) \right). \nonumber
\end{align}
We also have $ Z^{(r)'}(c) = rW^{(r)}(c)$ and $\frac{W^{(r)'}(c)}{W^{(r)}(c)} = \lambda(c,r)$. Therefore,
\begin{align}\label{niewiem1}
	 &\frac{\partial}{\partial \overline{x}} \hat{V}(x, \overline{x}) \big|_{x = \overline{x}} = \left( K - e^{\overline{x}-c} \right) \left[ -Z^{(r)'}(c) + \frac{Z^{(r)}(c)}{W^{(r)}(c)}W^{(r)'}(c) \right] \\& -e^{\overline{x}}\left(1 - \lambda(c,r)\right) - \lambda(c,r)\Delta(c,r)K.
\end{align}
Further algebraic simplifications give
\begin{align}
&\frac{\partial}{\partial \overline{x}} \hat{V}(x, \overline{x}) \big|_{x = \overline{x}}
= \left( K - e^{\overline{x}-c} \right) \left[ -Z^{(r)'}(c) + \frac{Z^{(r)}(c)}{W^{(r)}(c)}W^{(r)'}(c) \right] \nonumber\\ &- e^{\overline{x}}\left(1 - \lambda(c,r)\right) - \frac{\sigma^2}{2}K\left( \lambda(c,r)W^{(r)'}(c) - W^{(r)''}(c)\right) \nonumber\\
	&= K\left(  \frac{Z^{(r)}(c)}{W^{(r)}(c)}W^{(r)'}(c) - rW^{(r)}(c) - \frac{\left(\sigma W^{(r')}(c)\right)^2}{2W^{(r)}(c)} + \frac{\sigma^2}{2}W^{(r)''}(c) \right) \nonumber\\
	& + e^{\overline{x}}\left( rW^{(r)}(c)e^{-c} - \frac{Z^{(r)}(c)}{W^{(r)}(c)}W^{(r)'}(c)e^{-c} -1 + \frac{W^{(r)'}(c)}{W^{(r)}(c)} \right). \nonumber
\end{align}

We will prove that both brackets are equal to zero. Note that $Z^{(r)}(c) = \frac{1}{1-\gamma}\left( e^{\gamma} c - \gamma e^c \right)$ and observe that:
\begin{align}\label{niewiem2}
	  &\frac{Z^{(r)}(c)}{W^{(r)}(c)}W^{(r)'}(c) - rW^{(r)}(c) - \frac{\left(\sigma W^{(r')}(c)\right)^2}{2W^{(r)}(c)} + \frac{\sigma^2}{2}W^{(r)''}(c)  \\
	&=  \frac{1}{1-\gamma}\left( e^{\gamma c} - \gamma e^c \right)\frac{e^{c} - \gamma e^{\gamma c}}{e^{c} - e^{\gamma c}} - \frac{r}{r + \frac{\sigma^2}{2}}\left( e^{c} - e^{\gamma c} \right) - \frac{\frac{\sigma^2}{2}}{r + \frac{\sigma^2}{2}}\left( \frac{(e^{c} - \gamma e^{\gamma c})^2}{e^{c} - e^{\gamma c}} \right) \nonumber\\
	&=  \frac{1}{1-\gamma}\left( e^{\gamma c} - \gamma e^c \right)\frac{e^{c} - \gamma e^{\gamma c}}{e^{c} - e^{\gamma c}} + \frac{\gamma}{1-\gamma}\left( e^{c} - e^{\gamma c} \right) - \frac{1}{1-\gamma}\left( \frac{(e^{c} - \gamma e^{\gamma c})^2}{e^{c} - e^{\gamma c}} \right)  \nonumber\\ &=  -\frac{1}{r + \frac{\sigma^2}{2}}(e^c - e^{\gamma c}) + \frac{(e^c - e^{\gamma c})^{-1}}{1-\gamma } \nonumber\\
	&\times \left[ e^{c(1+\gamma^2)} - \gamma(e^{2c} + e^{2\gamma c}) - e^{2c} - \gamma^2e^{2\gamma c} + 2\gamma e^{c(1+\gamma)} + e^{2c} + \gamma^2e^{2\gamma c} - (1+\gamma^2)e^{c(1+\gamma)} \right] \nonumber\\
	&= \frac{\gamma(e^c - e^{\gamma c})^{-1}}{1-\gamma } \left( e^{2c} + e^{2\gamma c} - 2e^{c(1+\gamma)} \right)  -\frac{1}{r + \frac{\sigma^2}{2}}(e^c - e^{\gamma c}) \nonumber\\
	&= (e^c - e^{\gamma c})\frac{1}{1-\frac{1}{\gamma}} - \frac{1}{1-\frac{1}{\gamma}}(e^c - e^{\gamma c}) = 0. \nonumber
\end{align}
Additionally,
\begin{align}\label{niewiem3}
	&rW^{(r)}(c)e^{-c} - \frac{Z^{(r)}(c)}{W^{(r)}(c)}W^{(r)'}(c)e^{-c} -1 + \frac{W^{(r)'}(c)}{W^{(r)}(c)} = \\
	& = \frac{1}{W^{(r)}(c)} \left[ r(W^{(r)}(c))^2e^{-c} - Z^{(r)}(c)W^{(r)'}(c)e^{-c} -W^{(r)}(c) + W^{(r)'}(c) \right] \nonumber\\
	&=\frac{C e^{-c}}{W^{(r)}(c)}\left[ rC(e^c - e^{\gamma c})^2 - \frac{1}{1-\gamma}(e^{\gamma c} - \gamma e^c)(e^c - \gamma e^{\gamma c}) - e^c(e^c - e^{\gamma c} - e^c + \gamma e^{\gamma c} )  \right] \nonumber \\
	& = \frac{C e^{-c}}{W^{(r)}(c)} \left[ \frac{1}{1-\gamma}\left( \gamma e^{2c} + \gamma e^{2\gamma c} - e^{c(1+\gamma)} - \gamma(e^{2c} + e^{2\gamma c} - 2e^{c(1+\gamma)}) \right) + (1-\gamma)e^{c(1+\gamma)} \right] = 0. \nonumber
\end{align}

Combining (\ref{niewiem1}), (\ref{niewiem2}) and (\ref{niewiem3}) we get $	\frac{\partial}{\partial \overline{x}} \hat{V}(x, \overline{x}) \big|_{x = \overline{x}} = 0$.

This completes the proof.

\section{Numerical analysis}\label{sec:num}

In this section, we explore several properties of the options capped by drawdown. First, Figure \ref{Smooth_paste} illustrates the smooth-paste condition defined by equations (\ref{HJB6}) and (\ref{HJB7}).
The parameters are selected such that $\overline{x} < a^*+c$, which allows us to observe the occurrence of a drawdown event.

Next, Figure \ref{Price} shows how the option price depends on the initial values $X_0=x$ and $\overline{X}_0=\overline{x}$.
Notably, when $\overline{x}$ exceeds $\log(K) + c$, the option becomes worthless unless $x < \log(K)$. This is because, at such a level of $\overline{x}$, the stock price cannot reach the strike price before the option is terminated by the drawdown trigger. Of course, one can hypothetically consider a pair $(x, \overline{x})$ such that the difference between them is greater than $c$. In that case, the option is immediately exercised and its price is equal to the immediate payoff. For $\overline{x} < \log(K) + c$, the plot shows a smooth paste of the price function to the payoff function.

Finally, we perform a sensitivity analysis of both the stopping barrier and the option price with respect to the volatility $\sigma$ and the risk-free rate $r$.
In Figure \ref{Sensis_a}, we analyse the optimal barrier $e^{a^*}$ of the underlying asset price process. It is evident that the barrier increases with higher $r$ and lower $\sigma$, indicating that an increase in the drift parameter of $X_t$ leads to an upward shift in the barrier.

Similarly, Figure \ref{Sensis_v} presents the sensitivities of the option price. In contrast to the previous chart, we can see that the function increases with a higher interest rate and lower volatility. This behaviour is intuitive: A greater $\sigma$ leads to higher uncertainty for the seller, which has the effect on the risk premium. On the other hand, the decrease in $r$ translates to discount factors closer to 1, and consequently to the higher present value of future payouts.

The relation between Figures \ref{Sensis_a} and \ref{Sensis_v} is also consistent with theoretical expectations. The highest option price coincides with the lowest optimal stopping barrier, since reaching a lower asset price can lead to a greater payout than if the contract were exercised earlier.

\begin{figure}
	\includegraphics[width=\textwidth]{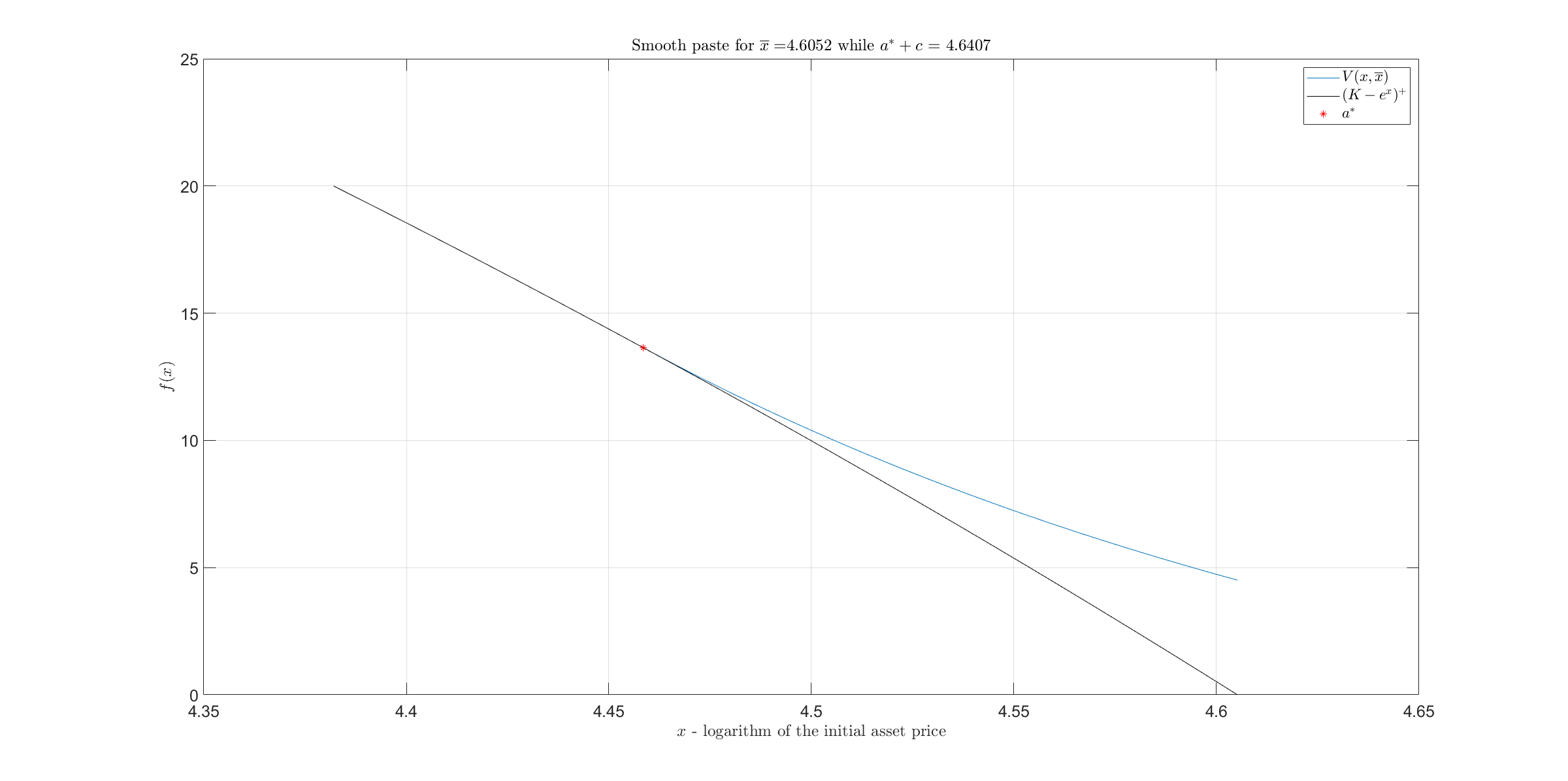}
	\caption{Smooth paste of the option price $V$ and the option payoff. Parameters of the model: $r = 0.1,\ \sigma = 0.2,\ e^c = 1.2,\ e^{\overline{x}} = K = 100$.}
	\label{Smooth_paste}
\end{figure}

\begin{figure}
	\includegraphics[width=\textwidth]{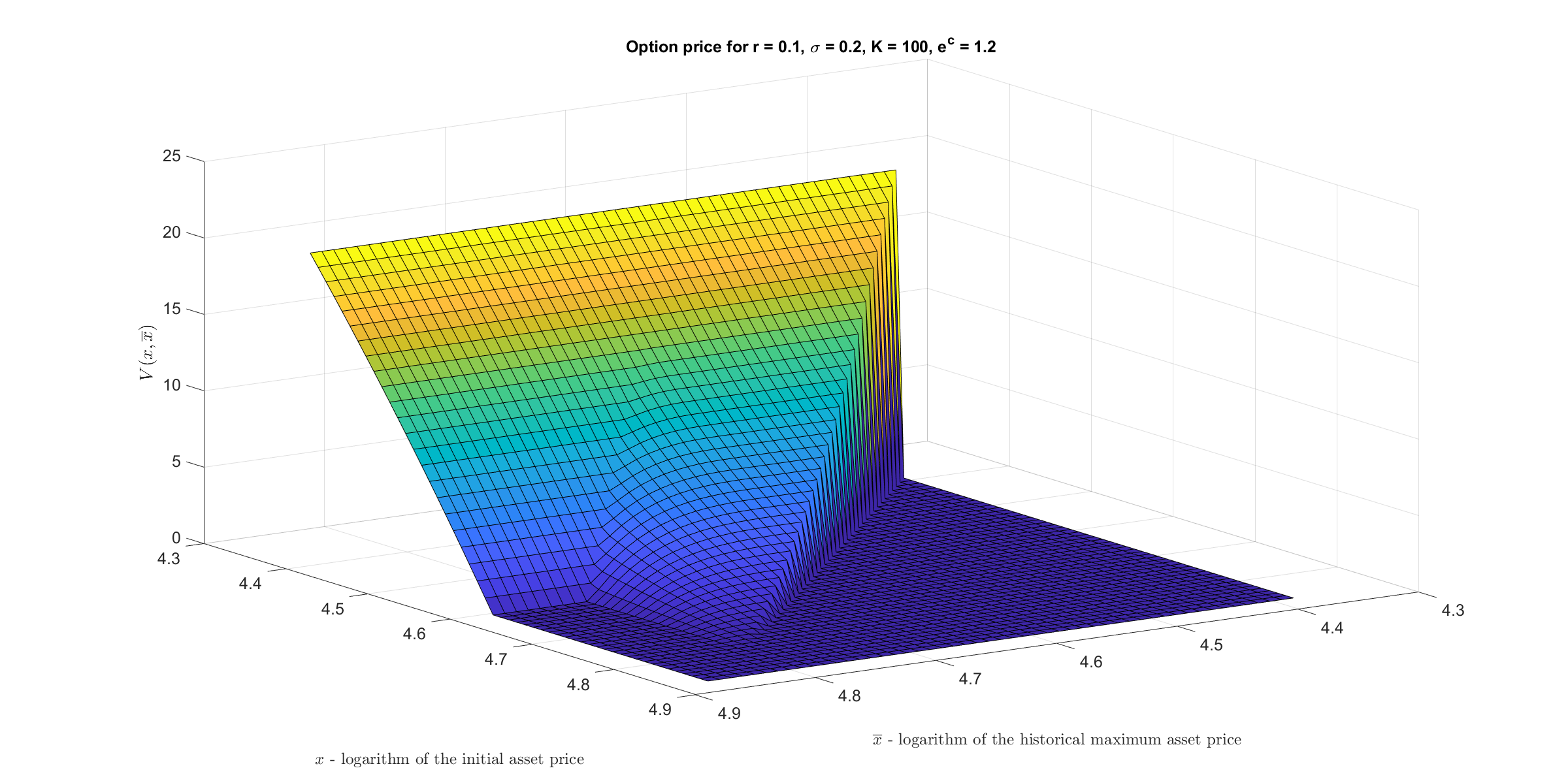}
	\caption{Option price depending on $x$ and $\overline{x}$. Note that the function is not defined for $\overline{x} < x$.}
	\label{Price}
\end{figure}

\begin{figure}
	\includegraphics[width=\textwidth]{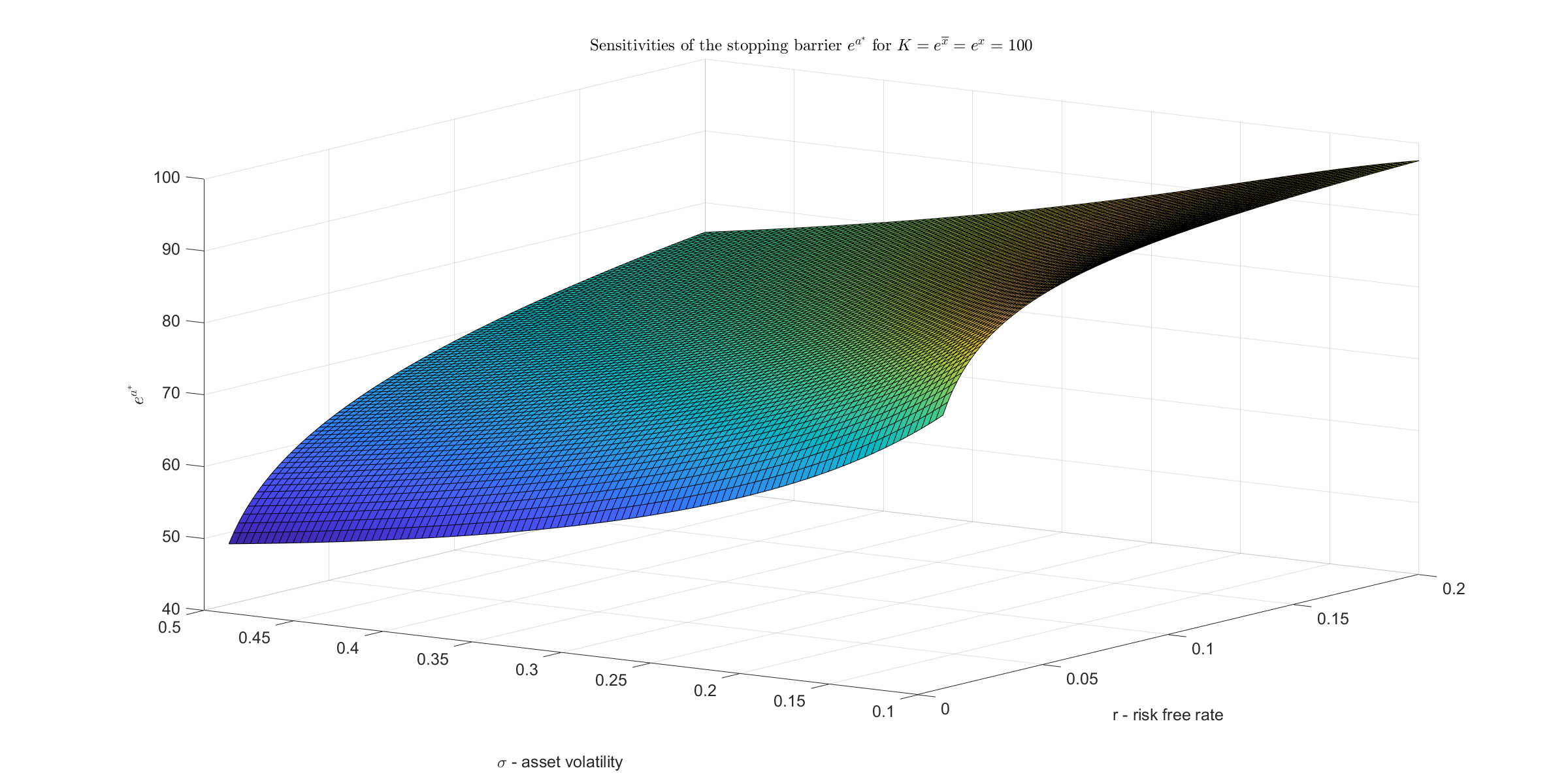}
	\caption{Stopping barrier $e^{a^*}$ of underlying asset price process $S_t$ depending on the risk-free rate $r$ and the volatility $\sigma$.}
	\label{Sensis_a}
\end{figure}

\begin{figure}
	\includegraphics[width=\textwidth]{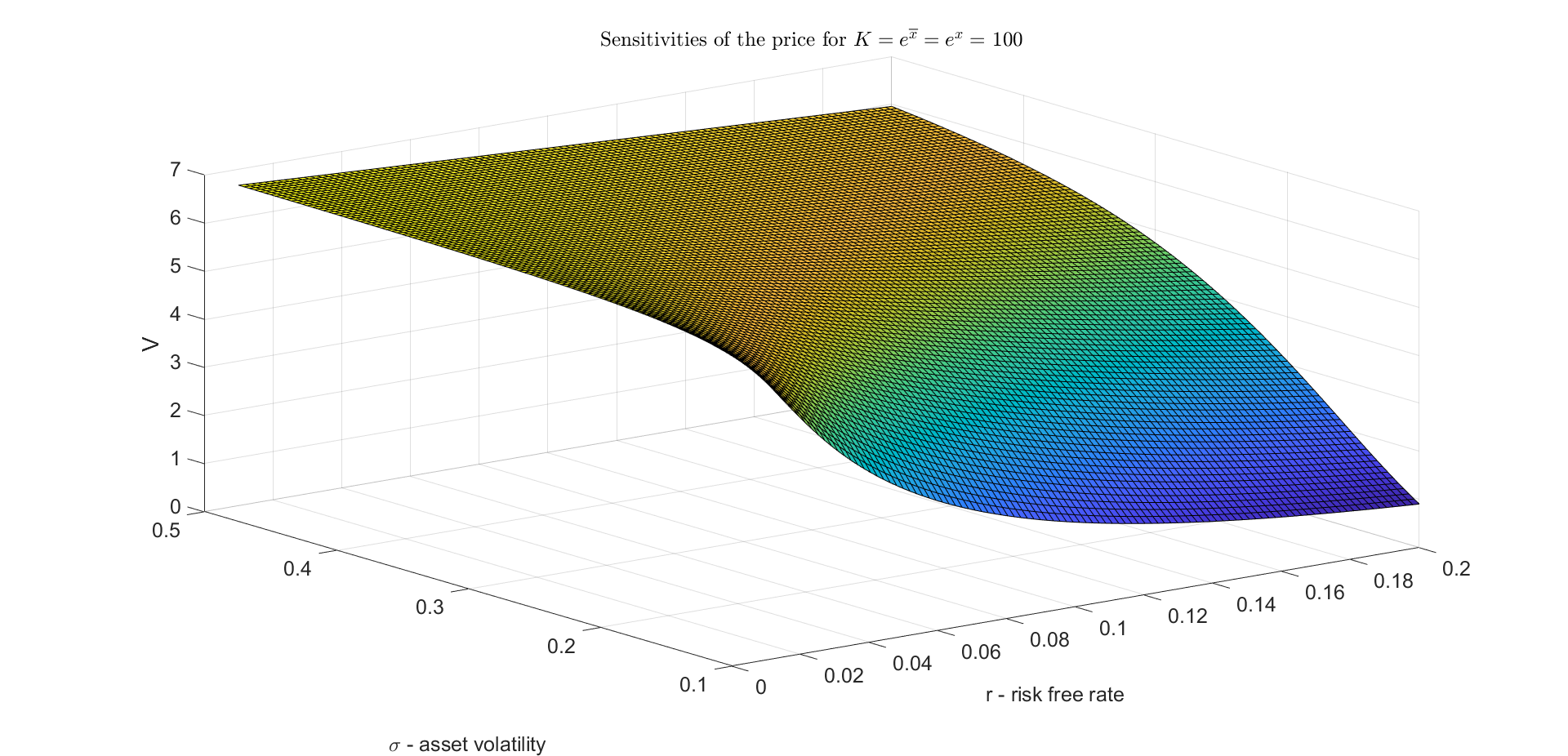}
	\caption{Sensitivities of the option price depending on the risk-free rate $r$ and the volatility $\sigma$.}
	\label{Sensis_v}
\end{figure}

{\bf Data availability statement}\\
This paper has no associated data.


\newpage

\begin{thebibliography}{10}
\expandafter\ifx\csname url\endcsname\relax
  \def\url#1{\texttt{#1}}\fi
\expandafter\ifx\csname urlprefix\endcsname\relax\def\urlprefix{URL }\fi
\expandafter\ifx\csname href\endcsname\relax
  \def\href#1#2{#2} \def\path#1{#1}\fi


\bibitem{Florin}
Avram, F., Chan, T.,  and Usabel, M. (2002)
On the valuation of constant barrier options under
spectrally one-sided exponential L\'evy models and
Carr’s approximation for American puts.
{\it Stoch. Process. Appl.} {\bf 100}, 75–107.

\bibitem{BieleckiRutkowski}
Bielecki, T.R. and Rutkowski, M. (2014)
{\it Credit Risk: Modeling, Valuation and Hedging}.
Springer.

\bibitem{cap} Broadie, M. and Detemple, J. (1995) American Capped Call Options on Dividend-Paying Assets. {\it The Review of Financial Studies} {\bf 8(1)} : 161--91.

\bibitem{carr} Carr, P. (1998) Randomization and the American Put. {\it The Review of Financial Studies} {\bf 11(3), 597-626}.

\bibitem{CZH}
Carr, P., Zhang, H. and Hadjiliadis, O. (2011) Maximum drawdown insurance. {\it International
Journal of Theoretical and Applied Finance} {\bf 14(8)}, 1--36.

\bibitem{EisenbaumKyprianou}
Eisenbaum, N. and Kyprianou, A. (2008)
On the parabolic generator of a general one-dimensional L\'evy
process, {\it Elec. Comm. Probab.}, {\bf (13)}, 198--208.

\bibitem{time_constraints} Egloff, D., Farkas, W. and Leippold, F. M. (2009) American Options with Stopping Time Constraints.
{\it Applied Mathematical Finance} {\bf 16(3)}, 287--305. 

\bibitem{gapeev1} Gapeev, P. and Al Motairi, H (2018) Perpetual American defaultable options in models with random dividends and partial information. {\it Risks} {\bf 6}, 127.

\bibitem{algo} Gapeev, P., Li, L. and Wu, Z. (2020) Perpetual American Cancellable Standard Options in Models with Last Passage Times. {\it Algorithms} {\bf 14}, 3.

\bibitem{1st_pass} Gapeev, P. and Al Motairi, H. (2022) Discounted optimal stopping problems in first-passage time models with random thresholds. {\it J. Appl. Probab.} {\bf 59}, 714--33.

\bibitem{GZ}
Grossman, S. J. and Zhou, Z. (1993) Optimal investment strategies for controlling drawdowns.
{\it Mathematical Finance} {\bf 3(3)}, 241--276.

\bibitem{kifer} Kifer, Y. (200) Game options. {\it Finance and Stochastics} {\bf 4}, 443--463.

\bibitem{MA}
Magdon-Ismail, M. and Atiya, A. (2004) Maximum drawdown. {\it Risk} {\bf 17(10)}, 99--102.

\bibitem{MP} Mijatovic, A. and Pistorius M.R. (2012) On the drawdown of completely asymmetric L\'evy processes. {\it Stochastic Processes and their Applications} {\bf 122(11)}, 3812-3836.

\bibitem{meyer} Meyer, G. H.  (2016) A PDE View of Games Options. {\it Research Paper Series} {\bf 369}, Quantitative Finance Research Centre .

\bibitem{ott} Ott, C. (2013) Optimal Stopping Problems for the Maximum Process. PhD thesis, University of Bath.

\bibitem{ZP&JT2}
Palmowski, Z. and Tumilewicz, J. (2018)
Pricing insurance drawdowns-type contracts with underlying L\'evy assets.
{\it Insurance: Mathematics and Economics} {\bf 79}, 1--14.

\bibitem{ZP&JT}
Palmowski, Z. and Tumilewicz, J. (2020)
Fair valuation of L\'evy-type drawdown-drawup contracts with general insured and penalty functions.
{\it Applied Mathematics and Optimization} {\bf 81}, 301--347.

\bibitem{PS} Peskir, G. and Shiryaev, A.N. (2006) Optimal Stopping and Free-Boundary Problems; Birkh\"{a}user: Basel, Switzerland.

\bibitem{drawdownup2}
Pospisil, L., Vecer, J., and Hadjiliadis, O. (2009) Formulas for stopped diffusion processes with stopping times based on drawdowns and drawups,
{\it Stoch. Process. Appl.} {\bf 119(8)}, 2563--2578.

\bibitem{PV}
Pospisil, L. and Vecer, J. (2010) Portfolio sensitivities to the changes in the maximum and
the maximum drawdown. {\it Quantitative Finance} {\bf 10(6)}, 617--627.

\bibitem{Sorn}
Sornette, D. (2003) {\it Why Stock Markets Crash: Critical Events in Complex Financial Systems.}
Princeton University Press.

\bibitem{trabelsi} Trabelsi, F. (2011) Asymptotic Behavior of Random Maturity American Options. {\it IAENG International Journal of Applied Mathematics}, {\bf 41}, 2.


\bibitem{Vec1}
Vecer, J. (2006) Maximum drawdown and directional trading. {\it Risk} {\bf 19(12)}, 88--92.

\bibitem{Vec2}
Vecer, J. (2007) Preventing portfolio losses by hedging maximum drawdown. {\it Wilmott} {\bf 5(4)},
1--8.

\bibitem{random_put} Wu, Z. and Li, L. (2022) The American Put Option with a Random Time Horizon. arXiv:2211.13918.

\bibitem{russian} Wu, Z. and Li, L. (2022) The Russian Option with A Random Time Horizon. arXiv:2211.13917v1.


\bibitem{drawdownup1}
Zhang, H. and Hadjiliadis, 0. (2009) Formulas for the Laplace transform of stopping times
based on drawdowns and drawups. http://arXiv.org/pdf/0911.1575.

\bibitem{olympia}
Zhang, H., Leung, T. and Hadjiliadis, O. (2013) Stochastic modeling and fair valuation of drawdown insurance. {\it Insurance: Mathematics and Economics} {\bf 53}, 840--850.


\bibitem{Zaevski2} Zaevski, T. (2020) Discounted perpetual game call options. {\it Chaos, Solitons \& Fractals} {\bf 131}, 109503.

\bibitem{Zaevski2b} Zaevski, T. (2020) Discounted perpetual game put options. {\it Chaos, Solitons \& Fractals} {\bf 137}, 109858.

\bibitem{Mladen}
Zaevskii, T., Kounchev, O., and Savov, M. (2019)
Two frameworks for pricing defaultable derivatives.
{\it Chaos, Solitons \& Fractals} {\bf 123}, 309--319.

\end{thebibliography}
\end{document}